\documentclass[10pt,conference]{IEEEtran}
\usepackage{fixltx2e}
\usepackage{cite}
\usepackage{url}
\usepackage{mathrsfs}
\usepackage{color}
\usepackage{float}
\usepackage[caption=false]{subfig}

\ifCLASSINFOpdf
\usepackage[pdftex]{graphicx}
\graphicspath{{Figs/}}
\DeclareGraphicsExtensions{.pdf,.jpeg,.png}
\else
\fi

\usepackage[cmex10]{amsmath}
\usepackage{amsmath}
\usepackage{amssymb}
\usepackage{amsthm}
\usepackage{amsfonts}
\usepackage{bm}
\usepackage{xfrac}
\usepackage{empheq}
\usepackage[normalem]{ulem} 
\usepackage{soul} 
\usepackage{mathtools}
\DeclarePairedDelimiter\ceil{\lceil}{\rceil}
\DeclarePairedDelimiter\floor{\lfloor}{\rfloor}

\newtheorem{theorem}{Theorem}

\newtheorem{example}{Example}

\newtheorem{remark}{Remark}
\theoremstyle{definition}

\usepackage[a4paper,bindingoffset=0.2in,%
left=1in,right=1in,top=1.02in,bottom=1.03in,%
footskip=.25in]{geometry}
\graphicspath{{figs/}}

\begin{document}
	\newgeometry{left=0.7in,right=0.7in,top=.72in,bottom=1.04in}
	\title{Improving Achievability of Cache-Aided Private Variable-Length Coding with Zero Leakage}
\vspace{-5mm}
\author{
		\IEEEauthorblockN{Amirreza Zamani, Mikael Skoglund \vspace*{0.5em}
			\IEEEauthorblockA{\\
                              Division of Information Science and Engineering, KTH Royal Institute of Technology \\
				Email: \protect amizam@kth.se, skoglund@kth.se }}
		}
	\maketitle
%
\begin{abstract} 
	A statistical cache-aided compression problem with a privacy constraint is studied, where a server has access to a database of $N$ files, $(Y_1,...,Y_N)$, each of size $F$ bits and is linked through a shared channel to $K$ users, 
	where each has access to a local cache memory of size $MF$ bits. During the placement phase, the server fills the users' caches without prior knowledge of their demands, while the delivery phase takes place after the users send their demands to the server. 
	We assume that each file in database $Y_i$ is arbitrarily correlated with a private attribute $X$, and an adversary is assumed to have access to the shared channel. The users and the server have access to a shared key $W$.  
	The goal is to design the cache contents and the delivered message $\cal C$ 
	such that the average length of $\mathcal{C}$ is minimized, while satisfying: 
i. The response $\cal C$ does not reveal any information about $X$, i.e., $I(X;\mathcal{C})=0$; ii. User $i$ can decode its demand, $Y_{d_i}$, by using the shared key $W$, $\cal C$, and its local cache $Z_i$. 
	 In a previous work, we have proposed a variable-length coding scheme that combines privacy-aware compression with coded caching techniques. 
	 In this paper, we propose a new achievability scheme using minimum entropy coupling concept and a greedy entropy-based algorithm. We show that the proposed scheme improves the previous results. Moreover, considering two special cases we improve the obtained bounds using the common information concept. 
\end{abstract}
\begin{IEEEkeywords}
	Cache-aided networks, private variable-length coding, minimum entropy functional representation. 
	\end{IEEEkeywords}
\section{Introduction}
 We consider the same scenario as in \cite{amircache} illustrated in Fig. \ref{wii}, in which a server has access to a database consisting of $N$ files $Y_1,..,Y_N$. Each file, of size $F$ bits, is sampled from the joint distribution $P_{XY_1\cdot Y_N}$, where $X$ denotes the private attribute. We assume that the server knows the realization of the private variable $X$ as well. The server is linked to $K$ users over a shared channel, where user $i$ has access to a local cache memory of size $MF$ bits. Furthermore, we assume that the server and the users have access to a shared key denoted by $W$, of size $T$. Similar to \cite{maddah1}, the system works in two phases: the placement and delivery phases. In the placement phase, the server fills the local caches using the database. After the placement phase, let $Z_k$ denote the content of the local cache of user $k$, $k\in[K]\triangleq\{1,..,K\}$. In the delivery phase, first the users send their demands to the server, where $d_k\in[N]$ denotes the demand of user $k$. The server sends a response, denoted by $\mathcal{C}$, over the shared channel to satisfy all the demands, simultaneously. We assume that an adversary has access to the shared link as well, and uses $\cal C$ to extract information about $X$. However, the adversary does not have access to the local cache contents or the shared secret key. 
 As argued in \cite{amircache}, since the files in the database are all correlated with the private latent variable $X$, the coded caching and delivery techniques introduced in \cite{maddah1} do not satisfy the privacy requirement. Similar to \cite{amircache}, the goal of the cache-aided private delivery problem is to find a response $\mathcal{C}$ with minimum possible average length that satisfies a certain privacy constraint and the zero-error decodability constraint of users. Similar to \cite{amircache}, we consider the worst case demand combinations $d=(d_1,..,d_K)$ to construct $\cal C$, and the expectation is taken over the randomness in the database. We consider a perfect privacy constraint, i.e., we require $I(\mathcal {C};X)=0$. Let $\hat{Y}_{d_k}$ denote the decoded message of user $k$ using $W$, $\cal C$, and $Z_k$. User $k$ should be able to decode $Y_{d_k}$ reliably, i.e., $\mathbb{P}{\{\hat{Y}_{d_k}\neq Y_{d_k}\}}=0$, $\forall k\in[K]$. 
  In \cite{amircache}, we have utilized techniques used in privacy mechanisms, data compression, and cache design and coded delivery problems, and combine them to build such a code. 
  In particular, we have used data compression techniques employed in \cite{kostala} and caching design techniques in \cite{maddah1}. In this work, to build $\cal C$, we use the minimum entropy coupling concept and a greedy entropy-based algorithm that are studied in \cite{kocaoglu2017entropic, compton2023minimum, shkel2023information}. 
  We compare the new proposed scheme with the existing one in \cite{amircache} and show that the proposed achievable scheme can significantly improve the previous result. 
\begin{figure}[]
	\centering
	\includegraphics[scale = .3,trim={0 5cm 0 5cm}]{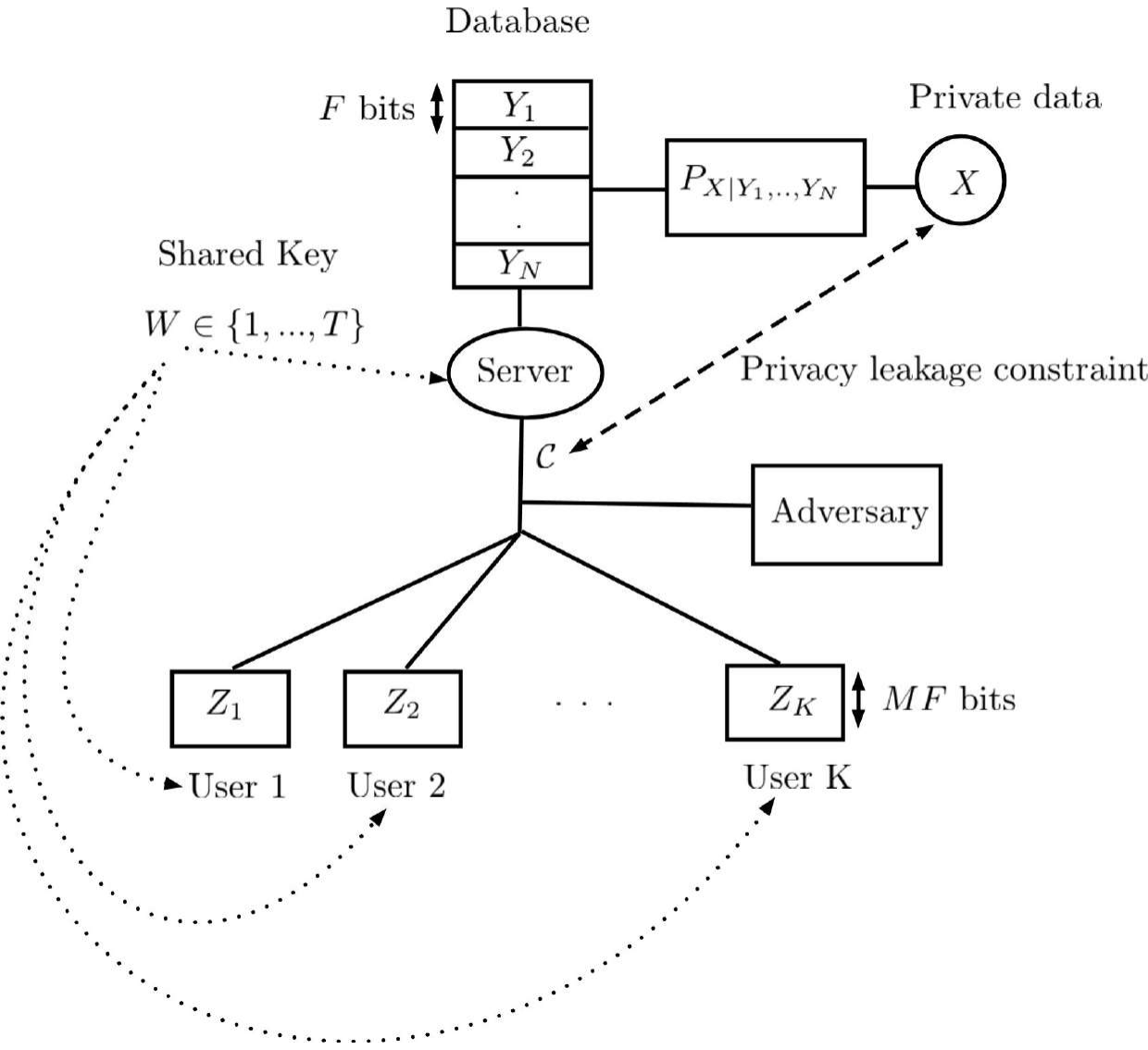}
	\caption{In this work a server wants to send a response over a shared link to satisfy users' demands, where the database is correlated with the private data. In the delivery phase, we hide the information about $X$ using one-time-pad coding and send the rest of response using a greedy entropy-based algorithm proposed in \cite{kocaoglu2017entropic}.}
	\label{wii}
\end{figure}
The privacy mechanism, caching and compression design problems have received increased attention in recent years
\cite{amircache, maddah1, kocaoglu2017entropic, compton2023minimum, shkel2023information, lim, wang33,  lu, denizjadid, shannon, borz, khodam, Khodam22,kostala, kostala2, asoodeh1,king3,9457633,e25040679,Total,shahab,makhdoumi}. 
Specifically, in \cite{maddah1}, a cache-aided network consisting of a single server connected
to several users equipped with local caches over a shared error-free link, is
considered, and the rate-memory trade-off has been characterized within a constant gap. This bound has been improved
for different scenarios in \cite{lim,wang33}. In particular, the exact rate-memory trade-off for uncoded placement has been characterized in \cite{lu}. 
A cache-aided coded content delivery problem is studied in \cite{denizjadid} where users have different distortion requirements.

Considering the compression problem with privacy constraints, a notion of perfect secrecy is introduced in \cite{shannon} by Shannon, where the private and public data are statistically independent. In the Shannon cipher system, one of $M$ messages is sent over a channel wiretapped by an eavesdropper, and it is shown that perfect secrecy is achievable if and only if the shared secret key length is at least $M$ \cite{shannon}. 


In \cite{borz}, the problem of privacy-utility trade-off considering mutual information both as measures of utility and privacy is studied. It is shown that under the perfect privacy assumption, the optimal privacy mechanism problem can be obtained as the solution of a linear program. 
In \cite{khodam}, the work \cite{borz} is generalized by relaxing the perfect privacy assumption allowing some small bounded leakage. 
This result is generalized to a non-invertible leakage matrix in \cite{Khodam22}.

In \cite{kostala}, the \emph{secrecy by design} concept has been introduced and is applied to privacy mechanism and lossless compression design problems. For the privacy problem, bounds on privacy-utility trade-off are derived by using the Functional Representation Lemma. These results are derived under the perfect privacy assumption.
The privacy problems considered in \cite{kostala} are generalized in \cite{king3} by relaxing the perfect privacy constraint. 
Moreover, the problems of fixed length and variable length compression with a certain privacy constraint have been studied in \cite{kostala} and upper and lower bounds on the average length of encoded message have been derived. 
A similar approach has been used in \cite{kostala2}, where in a lossless compression problem the relations between shared key, secrecy, and compression considering perfect secrecy, maximal leakage, secrecy by design, mutual information leakage and local differential privacy have been studied.
In \cite{amircache}, we have studied a cache-aided coded content delivery problem with a certain privacy constraint which is closely related to \cite{maddah1} and \cite{kostala}. We have provided upper and lower bounds on the average length of the server's response $\cal C$. Specifically, we have generalized the 
problem considered in \cite{maddah1} by considering correlation between the database and the private attribute, and we have used
variable-length lossless compression techniques as in \cite{kostala} to build the response $\cal C$ in the presence of an adversary. As argued in \cite{amircache}, for the achievability scheme we use a two-part code construction, which is based on the Functional Representation Lemma (FRL) and one-time pad coding to hide the information about $X$ and reconstruct the demanded files at user side.  

The main contribution of this work is to improve the acheivabilty scheme obtained in \cite{amircache}. Here, the acheivabilty scheme corresponds to the upper bound on the average length of the response $\cal C$. To do so we use a two-part code construction, which is based on the greedy entropy-based algorithm proposed in \cite{kocaoglu2017entropic} and one-time pad coding to hide the information about $X$ and reconstruct the demanded files at user side. Finally, considering two special cases, we improve the obtained bounds by using the common information concept. We show that when the size of the private data is large, the obtained bounds can be significantly improved using less shared key size. 

\section{system model and Problem Formulation} \label{sec:system}
In this work, the $i$-th file in the database is denoted by $Y_i$, which is of size $F$ bits, i.e., $\mathcal{Y}_i\in\{1,\ldots,2^F\}$ and $|\mathcal{Y}_i|=2^F$. Similar to \cite{amircache} we assume that $N\geq K$; however, the results can be generalized to other cases as well. The discrete random variable (RV) $X$ defined on the finite alphabet $\cal{X}$ describes the private attribute and is arbitrarily correlated with the files in the database $Y=(Y_1,\ldots,Y_N)$ where $|\mathcal{Y}|= |\mathcal{Y}_1|\times\ldots\times|\mathcal{Y}_N|=(2^F)^N$ and $\mathcal{Y}= \mathcal{Y}_1\times\ldots\times\mathcal{Y}_N$. 
The joint distribution of the private data and the database is shown by $P_{XY_1\cdot Y_N}$ and marginal distributions of $X$ and $Y_i$ by vectors $P_X$ and $P_{Y_i}$ defined on $\mathbb{R}^{|\mathcal{X}|}$ and $\mathbb{R}^{2^F}$ given by the row and column sums of $P_{XY_1\cdot Y_N}$. The relation between $X$ and $Y$ is given by the matrix $P_{Y_1\cdot Y_N|X}$ defined on $\mathbb{R}^{(2^F)^N\times|\mathcal{X}|}$.
Each user is equipped with a local cache of size $MF$ bits.
The shared secret key is denoted by the discrete RV $W$ defined on $\{1,\ldots,T\}$, and is assumed to be known by the server and the users, but not the adversary. Furthermore, we assume that $W$ is uniformly distributed and is independent of $(X,Y)$. 
Similar to \cite{maddah1}, we have $K$ caching functions to be used during the placement phase:
\begin{align}
\theta_k: [|\mathcal{Y}|] \rightarrow [2^{\floor{FM}}],\ \forall k\in[K], 
\end{align} 
such that
\begin{align}
Z_k=\theta_k(Y_1,\ldots,Y_N),\ \forall k\in[K].
\end{align} 
Let the vector $(Y_{d_1},\ldots,Y_{d_K})$ denote the demands sent by the users at the beginning of the delivery phase, where $(d_1,\ldots,d_K)\in[N]^K$.  
A variable-length prefix-free code with a shared secret key of size $T$ consists of mappings:
\begin{align*}
&(\text{encoder}) \ \mathcal{C}: \ [|\mathcal{Y}|]\times [T]\times[N]^K \rightarrow \{0,1\}^*\\
&(\text{decoder}) \mathcal{D}_k\!: \! \{0,1\}^*\!\!\times\! [T]\!\times\! [2^{\floor{MF}}]\!\times\! [N]^K\!\!\!\rightarrow\! 2^F\!\!\!,\ \! \forall k\!\in\![K].
\end{align*}
The output of the encoder $\mathcal{C}(Y,W,d_1,\ldots,d_K)$ is the codeword the server sends over the shared link in order to satisfy the demands of the users $(Y_{d_1},\ldots,Y_{d_K})$. At the user side, user $k$ employs the decoding function $\mathcal{D}_k$ to recover its demand $Y_{d_k}$, i.e., $\hat{Y}_{d_k}=\mathcal{D}_k(Z_k,W,\mathcal{C}(Y,W,d_1,\ldots,d_K),d_1,\ldots,d_K)$.
The variable-length code $(\mathcal{C},\mathcal{D}_1,..,\mathcal{D}_K)$ is lossless if for all $k\in[K]$ we have
\begin{align}\label{choon}
\mathbb{P}(\mathcal{D}_k(\mathcal{C}(Y,W,d_1,\ldots,d_K),W,Z_k,d_1,\ldots,d_K)\!=\!Y_{d_k})\!=\!1.
\end{align} 
In the following, we define perfectly private codes.
The code $(\mathcal{C},\mathcal{D}_1,\ldots,\mathcal{D}_K)$ is \textit{perfectly private} if
\begin{align}
I(\mathcal{C}(Y,W,d_1,\ldots,d_K);X)=0.\label{lashwi}
\end{align}
Let $\xi$ be the support of $\mathcal{C}(Y,W,d_1,\ldots,d_K)$, where $\xi\subseteq \{0,1\}^*$. For any $c\in\xi$, let $\mathbb{L}(c)$ be the length of the codeword. The lossless code $(\mathcal{C},\mathcal{D}_1,\ldots,\mathcal{D}_K)$ is \textit{$(\alpha,T,d_1,\ldots,d_K)$-variable-length} if 
\begin{align}\label{jojowi}
\mathbb{E}(\mathbb{L}(\mathcal{C}(Y,w,d_1,\ldots,d_K)))\!\leq\! \alpha,\ \forall w\!\in\!\![T]\ \text{and}\ \forall d_1,\ldots,d_K,
\end{align} 
and $(\mathcal{C},\mathcal{D}_1,\ldots,\mathcal{D}_K)$ satisfies \eqref{choon}.
Finally, let us define the set $\mathcal{H}(\alpha,T,d_1,\ldots,d_K)$ as follows:\\
$\mathcal{H}(\alpha,T,d_1,\ldots,d_K)\triangleq\{(\mathcal{C},\mathcal{D}_1,\ldots,\mathcal{D}_K): (\mathcal{C},\mathcal{D}_1,\ldots,\mathcal{D}_K)\ \text{is}\ \text{perfectly-private and}\\ (\alpha,T,d_1,\ldots,d_K)\text{-variable-length}  \}$. 
The cache-aided private compression design problems can be stated as follows
\begin{align}
\mathbb{L}(P_{XY_1\cdot Y_N},T)&=\!\!\!\!\!\inf_{\begin{array}{c} 
	\substack{d_1,\ldots,d_K,(\mathcal{C},\mathcal{D}_1,\ldots,\mathcal{D}_K)\in\mathcal{H}(\alpha,T,d_1,\ldots,d_K)}
	\end{array}}\alpha.\label{main1wi}
\end{align} 
\begin{remark}
	\normalfont 
	By letting $M=0$, $N=1$, and $K=1$, \eqref{main1wi} leads to the privacy-compression rate trade-off studied in \cite{kostala} and \cite{king3}, where upper and lower bounds have been derived.
\end{remark}
\begin{remark}
		\normalfont 
	Similar to \cite{amircache}, to design a code, we consider the worst case demand combinations. This follows since \eqref{jojowi} must hold for all possible combinations of the demands.
\end{remark}

 \section{Main Results}\label{sec:resul}
 In this section, we derive upper bounds on $\mathbb{L}(P_{XY_1\cdot Y_N},T)$ defined in \eqref{main1wi}. Providing new converse bounds is challenging and due to the limited space is left as a future work. Similar to \cite{amircache}, we employ the two-part code construction, which has been used in \cite{kostala}. We first encode the private data $X$ using a one-time pad \cite[Lemma~1]{kostala2}, then encode the RV found by the achievable scheme in \cite[Theorem~1]{maddah1} by using the greedy entropy based algorithm in \cite{kocaoglu2017entropic}. To do this, let us first recall the important results regarding upper and lower bounds on the minimum entropy coupling as obtained in \cite{kocaoglu2017entropic,compton2023minimum,shkel2023information}. Similar to \cite{shkel2023information}, for a given joint distribution $P_{XY}$ let the minimum entropy of functional representation of $(X,Y)$ be defined as
 \begin{align}
 H^*(P_{XY})=\!\!\!\!\!\inf_{\begin{array}{c} 
 	\substack{H(Y|X,U)=0,\ I(X;U)=0}
 	\end{array}}H(U).\label{minent}
 \end{align}
 \begin{remark}
 	\normalfont
 	By letting $\alpha=1$ in \cite[Definition 1]{shkel2023information}, it leads to the same problem in \eqref{minent}.
 \end{remark}
 \begin{remark}
 	\normalfont
 	As shown in \cite[Lemma 1]{shkel2023information}, the minimum entropy functional representation and the minimum entropy coupling are related functions. More specifically, $H^*(P_{XY})$ equals to the minimum entropy coupling of the set of PMFs $\{P_{Y|X=x_1},\ldots,P_{Y|X=x_n}\}$, where $\mathcal{X}=\{x_1,\ldots,x_n\}$.
 \end{remark}
Let $\mathcal{G}_S$ be the output of the greedy entropy-based algorithm which is proposed in \cite[Section 3]{kocaoglu2017entropic}, i.e., $H^*(P_{XY})\leq H(\mathcal{G}_S)$. More specifically, the corresponding algorithm aims to solve \eqref{minent} but does not achieve the optimal solution in general. Next, we recall a result obtained in \cite{compton2023minimum} which shows that $\mathcal{G}_S$ is optimal within $\frac{\log e}{e}\approx 0.53$ bits for $|\mathcal{X}|=2$ and is optimal within $\frac{1+\log e}{2}\approx 1.22$ bits for $|\mathcal{X}|>2$.  
 Let $U^*$ achieve the optimal solution of \eqref{minent}, i.e., $H(U^*)=H^*(P_{XY})$.
 \begin{theorem}\cite[Th. 3.4, Th. 4.1, Th. 4.2]{compton2023minimum}
 	Let $(X,Y)\sim P_{XY}$ and have finite alphabets. When $X$ is binary, we have
 	\begin{align}
 	H(Profile)&\leq H(U^*)\leq H(\mathcal{G}_S) \leq H(Profile)+\frac{\log e}{e}\nonumber \\ & \approx H(Profile)+0.53.
 	\end{align}
 	Moreover, for $|\mathcal{X}|>2$ we have
 	\begin{align}
 	H(Profile)&\leq H(U^*)\leq H(\mathcal{G}_S)\nonumber \\&\leq H(Profile)+\frac{1+\log e}{2}\nonumber \\ & \approx H(Profile)+1.22.
 	\end{align}
 	Here, $Profile$ corresponds to the profile method proposed in \cite[Section 3]{compton2023minimum}. 
 \end{theorem}
 Next, we present results on lower bounds on $H^*(P_{XY})$ obtained in a parallel work \cite{shkel2023information}. The lower bounds are obtained by using information spectrum and majorization concepts.
 \begin{theorem}\cite[Corollary 2, Th. 2]{shkel2023information}
 	Let $(X,Y)\sim P_{XY}$ and have finite alphabets. By letting $\alpha=1$ in \cite[Corollary 2, Th. 2]{shkel2023information}, we have
 	\begin{align}
 	H(\wedge_{x\in\mathcal{X}}P_{Y|x} )\leq H(Q^*)\leq H(U^*).
 	\end{align}
 	where $\wedge$ corresponds to the greatest lower bound with respect to majorization and $Q^*$ is defined in \cite[Lemma 3]{shkel2023information}.
 \end{theorem}
\begin{remark}
	\normalfont
	In contrast with \cite{compton2023minimum}, the lower bounds in \cite{shkel2023information} are obtained considering \textit{R{\'e}nyi} entropy in \eqref{minent}. In this paper, we consider \textit{Shannon} entropy which is a special case of \textit{R{\'e}nyi} entropy.
\end{remark}
\begin{remark}
	\normalfont
	As argued in \cite[Remark 1]{shkel2023information}, for $\alpha=1$ the (largest) lower bounds obtained in \cite{compton2023minimum} and \cite{shkel2023information} match. Thus, using Theorem 1, for binary $X$ we have
	\begin{align}
	H(Q^*)\leq H(U^*)\leq H(\mathcal{G}_S) \leq H(Q^*)+\frac{\log e}{e},
	\end{align}
	and for $|\mathcal{X}|>2$,
	\begin{align}
	H(Q^*)&\leq H(U^*)\leq H(\mathcal{G}_S)\leq H(Q^*)+\frac{1+\log e}{2}.
	\end{align}
	Moreover, in some cases the lower bound $H(Q^*)$ is tight, e.g., see \cite[Example 2]{shkel2023information}.
\end{remark}
As discussed in \cite{shkel2023information}, $Q^*$ can be obtained by a greedy construction. To do so, let $Q^*=(q_1^*,q_2^*,\ldots)$ with $q_1^*\geq q_2^*\geq...$, where $q_i=P(Q=q_i)$. Let $P_{Y|X}$ be a matrix with columns $P_{Y|X=x}$ where each is a conditional distribution vector and assume that each column has a descending order (re-order each column). Let $q_1^*= \min_{x\in\mathcal{X}}\{\max_{y\in{\mathcal{Y}}} P_{Y|X}(y|x)\}$. In other words, we choose the smallest number in the first row of the matrix $P_{Y|X}$. Next, we subtract $q_1^*$ from the first row and reorder each column and update the matrix. We then choose the smallest number from the first row of the updated matrix and represent it by $q_2^*$. We continue this procedure until the summation of $q_i^*$ reaches one. To see an example refer to \cite[Example 1]{shkel2023information}.   

 Next, we present a summary of the achievable scheme proposed in \cite[Theorem~1]{maddah1}. We first consider a cache size $M\in\{\frac{N}{K},\frac{2N}{K},\ldots,N\}$ and define $p\triangleq\frac{MK}{N}$, which is an integer. In the placement phase, each file, e.g., $Y_n$, $n\in[N]$, is split into $\binom{K}{p}$ equal size subfiles and labeled as follows
 \begin{align}
 Y_n=(Y_{n,\Omega}:\Omega\subset[K],|\Omega|=p).
 \end{align}
For all $n$, the server places $Y_{n,\Omega}$ in the cache of user $k$ if $k\in \Omega$. As argued in \cite{maddah1}, each user caches total of $N\binom{K-1}{p-1}\frac{F}{\binom{K}{p}}=MF$ bits, which satisfies the memory constraint with equality. In the delivery phase, 
the server sends the following message over the shared link
\begin{align}\label{cache1}
\mathcal{C}'\triangleq(C_{\gamma_1},\ldots,C_{\gamma_L}),
\end{align}
where $L=\binom{K}{p+1}$ and for any $i\in\{1,\ldots,L\}$, $\gamma_i$ is the $i$-th subset of $[K]$ with cardinality $|\gamma_i|=p+1$, furthermore,
\begin{align}\label{cache2}
C_{\gamma_i}\triangleq\oplus_{j\in \gamma_i} Y_{d_j,\gamma_i \backslash \{j\} },
\end{align}
 where $\oplus$ denotes bitwise XOR function. Note that $Y_{d_j,\gamma_i \backslash \{j\} }$ is the subfile that is not cached by user $j$, but is requested by it. In other words, considering each subset of $[K]$ with cardinality $|\gamma_i|=p+1$, using the message $C_{\gamma_i}$, each user, e.g., user $j$, is able to decode the subfile $Y_{d_j,\gamma_i \backslash \{j\} }$ that is not cached by it. Considering all the messages in \eqref{cache1} user $i$ can decode file $Y_{d_i}$ completely using $\mathcal{C}'$ and its local cache content $Z_i$. Note that each subfile $C_{\gamma_i}$ has size $\frac{F}{\binom{K}{p}}$ bits. As pointed out in \cite{maddah1}, for other values of $M$ we can use the memory-sharing technique. For more details see \cite[Proof of Theorem~1]{maddah1}.
 \subsection{New achievable scheme:}
 In this part, we present our achievable scheme which leads to upper bounds on \eqref{main1wi}. For simplicity let 
 $|\mathcal{C}'|$ be the cardinality of the codeword defined in \eqref{cache1} where $|\mathcal{C}'|= |\mathcal{C}_{\gamma_1}|\times\ldots\times|\mathcal{C}_{\gamma_L}|$. In the next result let $Q^*$ achieve the lower bound in Theorem 2 for the following problem 
  \begin{align}
 H^*(P_{X\mathcal{C}'})=\!\!\!\!\!\inf_{\begin{array}{c} 
 	\substack{H(\mathcal{C}'|X,U)=0,\ I(X;U)=0}
 	\end{array}}H(U),\label{minent2}
 \end{align}
 where in \eqref{minent}, $Y$ is substituted by $\mathcal{C}'$, i.e., $\mathcal{C}'\leftarrow Y$, and $\mathcal{C}'$ is as defined in \eqref{cache1}. Using Theorem 2 and Remark 6, $H^*(P_{X\mathcal{C}'})$ can be lower bounded by $H(Q^*)$ and upper bounded by $H(Q^*)+0.53$ when $|X|=2$ and by $H(Q^*)+1.22$ when $|X|>2$. For binary $X$ we have
 \begin{align}\label{kiun}
 H(Q^*)\leq H^*(P_{X\mathcal{C}'}) \leq H(Q^*)+\frac{\log e}{e},
 \end{align}
 and for $|X|>2$,
 \begin{align}\label{kiun1}
 H(Q^*)\leq H^*(P_{X\mathcal{C}'}) \leq H(Q^*)+\frac{1+\log e}{2}.
 \end{align}
 We emphasize that $Q^*$ that is used in \eqref{kiun} and \eqref{kiun1} is constructed using the greedy approach based on the matrix $P_{\mathcal{C}'|X}$. We use the same $Q^*$ for the following result.
 \begin{theorem}\label{th1}
 	Let RVs $(X,Y)=(X,Y_1,\ldots,Y_N)$ be distributed according to $P_{XY_1\cdot Y_N}$ supported on alphabets $\mathcal{X}$ and $\mathcal{Y}$, where $|\mathcal{X}|$ and $|\mathcal{Y}|$ are finite, and let the shared secret key size be $|\mathcal{X}|$, i.e., $T=|\mathcal{X}|$. Furthermore, let $M\in\{\frac{N}{K},\frac{2N}{K},\ldots,N\}$. Let $|X|=2$, we have
 	\begin{align}
 	\mathbb{L}(P_{XY},2)\leq \!H(Q^*)\!+\frac{\log e}{e}+\!2,\label{koonwi}
 	\end{align}
 	where $\mathcal{C}'$ is as defined in \eqref{cache1}. When $|X|>2$, we have
 	 \begin{align}
 	 \mathbb{L}(P_{XY},|\mathcal{X}|)\leq \!H(Q^*)\!+\frac{1+\log e}{2}+\!1+\!\ceil{\log (|\mathcal{X}|)},\label{koonwi2}
 	 \end{align}
 \end{theorem}
 \begin{proof}
 	The proof is similar to \cite{amircache} and the main difference is to use minimum entropy output of \eqref{minent2} instead of FRL that is used in two-part construction coding in \cite{amircache}.
 	In the placement phase, we use the same scheme as discussed before. In the delivery phase, we use the following strategy.
 	Similar to \cite{amircache}, we use two-part code construction to achieve the upper bounds. As shown in Fig. \ref{achieve}, we first encode the private data $X$ using one-time pad coding \cite[Lemma~1]{kostala2}, which uses $\ceil{\log(|\mathcal{X}|)}$ bits. 
 	Next, we produce $U$ based on greedy entropy-based algorithm proposed in \cite{kocaoglu2017entropic} which solves the minimum entropy problem in \eqref{minent2}, where $\mathcal{C}'$, defined in \eqref{cache1}, is the response that the server sends over the shared link to satisfy the users$'$ demands \cite{maddah1}. 
 	Thus, we have
 	\begin{align}
 	H(\mathcal{C}'|X,U)&=0,\label{kharkosde}\\
 	I(U;X)&=0,
 	\end{align}  
 	Note that in Remark 6 we substitute $\mathcal{G}_S$ by $U$ and for binary $X$ we have
 	\begin{align}\label{kiun2}
 	H(U) \leq H(Q^*)+\frac{\log e}{e},
 	\end{align}
 	and for $|X|>2$,
 	\begin{align}\label{kiun3}
 	H(U) \leq H(Q^*)+\frac{1+\log e}{2}.
 	\end{align}
 	 
 	Thus, we obtain \eqref{kiun} and \eqref{kiun1}.
Moreover, for the leakage constraint we note that the randomness of one-time-pad coding is independent of $X$ and the output of the greedy entropy-based algorithm $U$. 
 \begin{figure}[h]
 	\centering
 	\includegraphics[scale = .4,trim={0 9cm 0 9cm}]{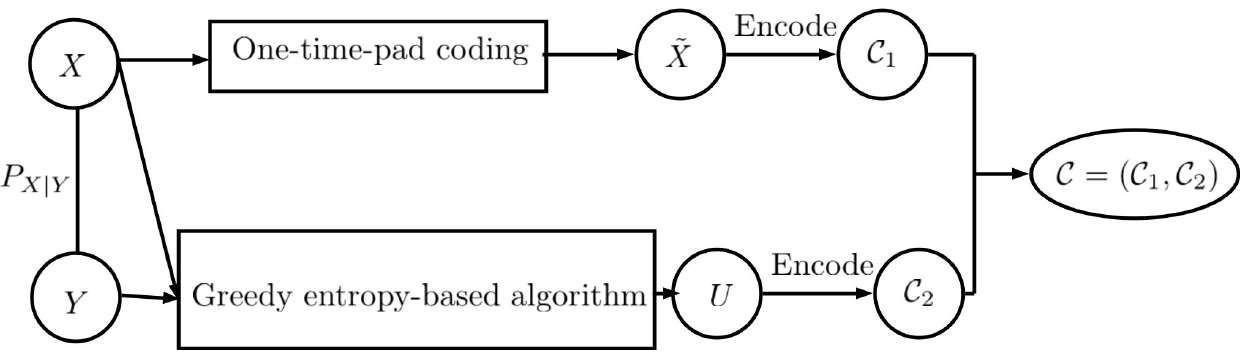}
 	\caption{Encoder design: illustration of the achievability scheme of Theorem \ref{th1}. Two-part code construction is used to produce the response of the server, $\mathcal{C}$. The server sends $\cal C$ over the channel, which is independent of $X$.}
 	\label{achieve}
 \end{figure}
As shown in Fig. \ref{decode}, at user side, each user, e.g., user $i$, first decodes $X$ using one-time-pad decoder. Then, based on \eqref{kharkosde} it decodes $\mathcal{C}'$ using $U$ and $X$. Finally, it decodes $Y_{d_i}$ using local cache $Z_i$ and the response $\mathcal{C}'$.  
\begin{figure}[h]
	\centering
	\includegraphics[scale = .4,trim={0 9cm 0 9cm}]{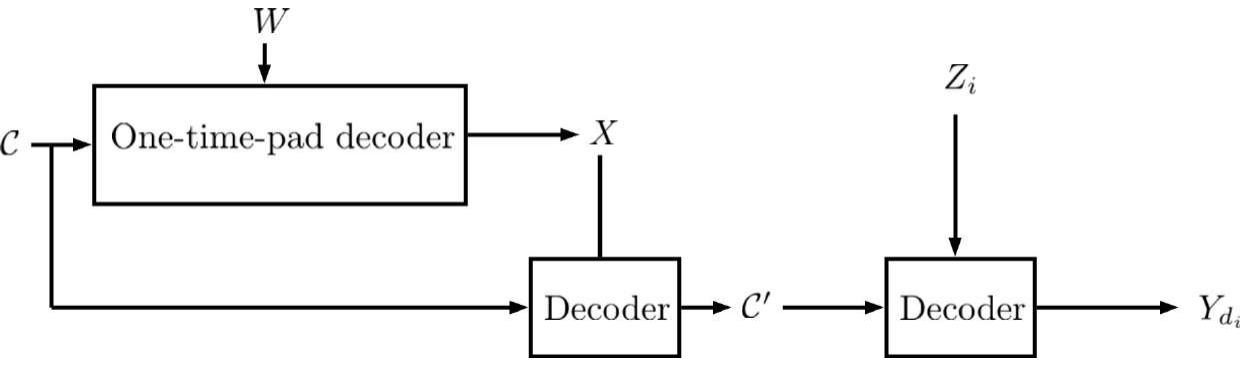}
	\caption{Illustration of the decoding process for the achievability scheme of Theorem \ref{th1}.}
	\label{decode}
\end{figure}
  \end{proof}
 \begin{remark}
 	\normalfont
 	As we mentioned earlier, the main difference between the present scheme and \cite{amircache} is to use greedy entropy-based algorithm which aims to minimize the output of FRL and is optimal within a constant gap.  
 \end{remark}
\begin{remark}
	The complexity of the algorithm to find $Q^*$ in Theorem \ref{th1} is linear in $|\mathcal{C}'|\times|\mathcal{X}|$. This can be shown by using \cite[Lemma 3]{shkel2023information}.
\end{remark}
\begin{remark}
	Although in Theorem~\ref{th1} we assume that $M\in\{\frac{N}{K},\frac{2N}{K},\ldots,N\}$, the results can be extended for other values of $M$ using the memory sharing technique of \cite[Theorem~1]{maddah1}.  
\end{remark}
\begin{remark}
	Similar to \cite{amircache}, we assume that the privacy leakage constraint needs to be fulfilled in the delivery phase. This assumption can be motivated since the placement phase occurs during the off-peak hours and we can assume that the adversary does not listen to the channel during that time. Considering the scenarios in the presence of the adversary during the placement phase, the server can employ the same strategy as used during the delivery phase. The server can fill the caches using the two-part code construction coding. 
\end{remark}
 Next we study a numerical example to better illustrate the achievable scheme in Theorem~\ref{th1}.
 \begin{example}(\cite[Example 1]{amircache})
 	Let $F=N=K=2$ and $M=1$. Thus, $Y_1=(Y_1^1,Y_1^2)$ and $Y_2=(Y_2^1,Y_2^2)$, where $Y_i^j\in\{0,1\}$ for $i,j\in\{1,2\}$. We assume that $X$ is the pair of first bits of the database, i.e., $X=(Y_1^1,Y_2^1)$, $Y_1$ and $Y_2$ are independent and have the following distributions
 	\begin{align*}
 	P(Y_{1}^1 =Y_{1}^2=0)&=P(Y_{1}^1 =Y_{1}^2=1)=\frac{1}{16},\\
 	P(Y_{1}^1 =1,Y_{1}^2=0)&=P(Y_{1}^1 =0,Y_{1}^2=1)=\frac{7}{16},\\
 	P(Y_{2}^1 =Y_{2}^2=0)&=P(Y_{2}^1 =Y_{2}^2=1)=\frac{1}{10},\\
 	P(Y_{2}^1 =1,Y_{2}^2=0)&=P(Y_{2}^1 =0,Y_{2}^2=1)=\frac{2}{5},
 	\end{align*}
 	In this case, the marginal distributions can be calculated as $P(Y_1^1=1)=P(Y_1^2=1)=P(Y_2^1=1)=P(Y_2^2=1)=\frac{1}{2}$. In the placement phase, the server fills the first local cache by the first bits of the database, i.e., $Z_1=\{Y_1^1,Y_2^1\}$ and the second local cache by the second bits, i.e., $Z_2=\{Y_1^2,Y_2^2\}$. In the delivery phase, assume that users $1$ and $2$ request $Y_1$ and $Y_2$, respectively, i.e., $Y_{d_1}=Y_1$ and $Y_{d_2}=Y_2$. In this case, $\mathcal{C}'=Y_1^2\oplus Y_2^1$, where $\mathcal{C}'$ is the server's response without considering the privacy constraint. Thus, $|\mathcal{C}'|=2$ and 
 	$
 	P_{\mathcal{C}'|X}=\begin{bmatrix}
 	\frac{7}{8} & \frac{1}{8} & \frac{1}{8} & \frac{7}{8}\\
 	\frac{1}{8} & \frac{7}{8} & \frac{7}{8} & \frac{1}{8}
 	\end{bmatrix}
 	$
 	Moreover, $Q^*$ has the following distribution $P_{Q^*}=[\frac{7}{8}, \frac{1}{8}]$, hence, $H(Q^*)=h(1/8)=0.5436$.
 	Next, we encode $X$ using $W$ as follows
 	\begin{align*}
 	\tilde{X}=X+W\ \text{mod}\ 4,
 	\end{align*}
 	where $W\perp X$ is a RV with uniform distribution over $\cal X$. To encode $\tilde{X}$ we use 2 bits. We then encode $\mathcal{C}'$ using greedy entropy-based algorithm. Let $U$ denote the output of the algorithm which satisfies \eqref{kiun3}. Let $\mathcal{C}_1$ and $\mathcal{C}_2$ describe the encoded $\tilde{X}$ and $U$, respectively. The server sends $\mathcal{C}=(\mathcal{C}_1,\mathcal{C}_2)$ over the shared link.
 	For this particular demand vector, using \eqref{koonwi2}, the average codelength is $4.7636$ bits. For the request vector $(Y_{d_1},Y_{d_2})=(Y_1,Y_2)$, the average length of the code is $4.7636$ bits to satisfy the zero leakage constraint. Thus, for $(Y_{d_1},Y_{d_2})=(Y_1,Y_2)$, we have
 	\begin{align*}
 	\mathbb{L}(P_{XY},4) \leq 4.7636\ \text{bits}.
 	\end{align*}
 	Using \cite{amircache}, for this particular demand vector we need $5$ bits.
 Both users first decode $X$ using $\tilde{X}$ and $W$, then decode $\mathcal{C}'=Y_1^2\oplus Y_2^1$ by using $X$ and $U$, since from FRL $\mathcal{C}'$ is a deterministic function of $U$ and $X$. User $1$ can decode $Y_1^2$ using $=Y_1^2\oplus Y_2^1$ and $Y_1^1$, which is available in the local cache $Z_1$, and user $2$ can decode $Y_2^1$ using $=Y_1^2\oplus Y_2^1$ and $Y_1^2$, which is in $Z_2$. Moreover, we choose $W$ to be independent of $X$ and $U$. As a result, $X$ and $(\tilde{X},U)$ become independent. Thus, $I(\mathcal{C};X)=0$, which means there is no leakage from $X$ to the adversary. 
 Next, assume that in the delivery phase both users request $Y_1$, i.e., $Y_{d_1}=Y_{d_2}=Y_1$. In this case, $\mathcal{C}'=Y_1^2\oplus Y_1^1$ with $|\mathcal{C}'|=2$. Using the same arguments we need $4.7636$ bits.
 Next, let $Y_{d_1}=Y_{d_2}=Y_2$. In this case, $\mathcal{C}'=Y_2^1\oplus Y_2^2$. In this case, $H(Q^*)=h(1/5)=0.7219$ and we need $4.9419$ bits. 
 Finally, let $Y_{d_1}=Y_2,\ Y_{d_2}=Y_1$. In this case, $\mathcal{C}'=Y_2^1\oplus Y_1^1$. Since $\mathcal{C}'$ is a function of $X$ it is enough to only send $X$ using on-time pad coding. Thus, for the request vector $Y_{d_1}=Y_2,\ Y_{d_2}=Y_1$, the average length of the code is $2$ bits to satisfy the zero leakage constraint.
 We conclude that in all cases we need less bits to send compared to \cite{amircache}, since by using \cite[Example 1]{amircache} we need 5 bits on average to send over the channel.   
 \end{example}
\subsection{Special case: improving the bounds using the common information concept}
In this section, we improve the bounds obtained in Theorem \ref{th1} considering a special case. To do so, let us recall the privacy mechanism design problems considered in \cite{king1} with zero leakage as follows
\begin{align}
g_{0}(P_{XY})&=\max_{\begin{array}{c} 
	\substack{P_{U|Y}:X-Y-U\\ \ I(U;X)=0,}
	\end{array}}I(Y;U),\label{maing}\\
h_{0}(P_{XY})&=\max_{\begin{array}{c} 
	\substack{P_{U|Y,X}: I(U;X)=0,}
	\end{array}}I(Y;U).\label{mainh}
\end{align} 
Finally, we define a set of joint distributions $\hat{\mathcal{P}}_{XY}$ as follows
\begin{align}
\hat{\mathcal{P}}_{XY}\triangleq \{P_{XY}:g_{0}(P_{XY})=h_{0}(P_{XY})\}.
\end{align}
As outlined in \cite[Lemma 1]{zero}, a sufficient condition to have $g_{0}(P_{XY})=h_{0}(P_{XY})$ is to have $C(X;Y)=I(X;Y)$, where $C(X,Y)$ denotes the common information between $X$ and $Y$, where common information corresponds to the Wyner \cite{wyner} or G{\'a}cs-K{\"o}rner \cite{gacs1973common} notions of common information. Moreover, a sufficient condition for $C(X;Y)=I(X;Y)$ is to let $X$ be a deterministic function of $Y$ or $Y$ be a deterministic function of $X$. In both cases, $C(X;Y)=I(X;Y)$ and $g_{0}(P_{XY})=h_{0}(P_{XY})$. For more detail see \cite[Proposition 6]{king3}. Moreover, in \cite[Lemma 2]{zero}, properties of the optimizers for $g_{0}(P_{XY})$ and $h_{0}(P_{XY})$ are obtained considering $P_{XY}\in \hat{\mathcal{P}}_{XY}$. It has been shown that the optimizer $U^*$ achieving $g_{0}(P_{XY})=h_{0}(P_{XY})$ satisfies
\begin{align}
H(Y|U^*,X)=0,\label{2}\\
I(X;U^*|Y)=0,\label{3}\\
I(X;U^*)=0.\label{4}
\end{align} 
Next, we recall the definitions of a set $\mathcal{U}^1(P_{XY})$ and a function $\mathcal{K}(P_{XY})$ in \cite{zero} as follows. 
\begin{align}
\mathcal{U}^1(P_{XY})&\triangleq \{U: U\ \text{satisfies \eqref{2}, \eqref{3}, \eqref{4}}\}\\
\mathcal{K}(P_{XY})&\triangleq \min_{U\in \mathcal{U}^1(P_{XY})} H(U).
\end{align} 
Noting that the function $\mathcal{K}(P_{XY})$ finds the minimum entropy of all optimizers satisfying $g_0(P_{XY})=h_0(P_{XY})$. A simple bound on $\mathcal{K}(P_{XY})$ has been obtained in \cite[Lemma 4]{zero}. Next, we define
 \begin{align}
A_{XY}&\triangleq \begin{bmatrix}
&P_{y_1}-P_{y_1|x_1} &\ldots & P_{y_{|\mathcal{Y}|}}-P_{y_{q}|x_1}\\
&\cdot &\ldots &\cdot\\
&P_{y_1}-P_{y_1|x_{t}} &\ldots & P_{y_{q}}-P_{y_{q}|x_{t}}
\end{bmatrix}\!\in\! \mathbb{R}^{t\times q},\\
b_{XY}&\triangleq \begin{bmatrix}
H(Y|x_1)-H(Y|X) \\
\cdot \\
H(Y|x_t)-H(Y|X)
\end{bmatrix}\in\mathbb{R}^{t},\ \bm{a}\triangleq \begin{bmatrix}
a_1 \\
\cdot \\
a_q
\end{bmatrix}\in\mathbb{R}^{q}.
\end{align}
where $t=|\mathcal{X}|$ and $q=|\mathcal{Y}|$. Noting that in \cite[Theorem 1]{zero}, bounds on $\mathcal{K}(P_{XY})$ and entropy of any $U\in\mathcal{U}^1$ have been obtained. Specifically, when $\text{rank}(A_{XY})=|\mathcal{Y}|$, the exact value of $\mathcal{K}(P_{XY})$ is obtained by solving simple linear equations in \cite[eq. (26)]{zero}. 
We emphasize that by using \cite{borz}, $g_0(P_{XY})$ can be obtained by solving a linear program in which the size of the matrix in the system of linear equations is at most $|\mathcal{Y}|\times\binom{|\mathcal{Y}|}{\text{rank}(P_{X|Y})}$ with at most $\binom{|\mathcal{Y}|}{\text{rank}(P_{X|Y})}$ variables. By solving the linear program as proposed in \cite{borz} we can find the exact value of $\mathcal{K}(P_{XY})$ and the joint distribution $P_{U|YX}$ that achieves it. The complexity of the linear program in \cite{borz} can grow faster than exponential functions with respect to $|\mathcal{Y}|$, however the complexity of the proposed method in \cite{zero} grows linearly with $|\mathcal{Y}|$. Thus, our proposed upper bound has less complexity compared to the solution in \cite{borz}.
The bounds on $\mathcal{K}(P_{XY})$ help us to obtain the next result. Next, we improve the bounds obtained in Theorem \ref{th1}.
\begin{theorem}\label{loo}
	Let RVs $(X,Y)=(X,Y_1,\ldots,Y_N)$ be distributed according to $P_{XY_1\cdot Y_N}$ supported on alphabets $\mathcal{X}$ and $\mathcal{Y}$, where $|\mathcal{X}|$ and $|\mathcal{Y}|$ are finite, and let the shared secret key size be $|\mathcal{X}|$, i.e., $T=|\mathcal{X}|$. Furthermore, let $M\in\{\frac{N}{K},\frac{2N}{K},\ldots,N\}$. Let $P_{X\mathcal{C}'}\in\hat{\mathcal{P}}_{X\mathcal{C}'}$ and let $q=|\mathcal{C}'|$ and $\beta=\log(\text{null}(P_{X|\mathcal{C}'})+1)$, where $\mathcal{C}'$ is defined in \eqref{cache1}. Then, we have 
	\begin{align}
	&\mathbb{L}(P_{X\mathcal{C}'},|\mathcal{X}|)
	\leq \mathcal{K}(P_{X\mathcal{C}'})+1+\ceil{\log(|\mathcal{X}|)} \label{log}\\
	&\leq H(\mathcal{C}'|X)\!\!+\!\!\!\!\!\!\!\!\!\!\!\!\!\!\!\!\!\!\!\!\max_{\begin{array}{c}
		\substack{a_i:A_{XY}\bm{a}=b_{XY},\bm{a}\geq 0,\\
			\sum_{i=1}^{q}\! P_{c'_i}a_i\leq \beta-H(\mathcal{C}'|X)} \end{array} }\!\!\sum_{i=1}^{q} \!\!P_{c'_i}a_i\!+\!1\!+\!\ceil{\log(|\mathcal{X}|)}\label{ass}\\&\leq \beta+1+\!\ceil{\log(|\mathcal{X}|)},\label{mass}
	\end{align}
	where $c'_i$ is the $i$-th element (alphabet) of $\mathcal{C}'$.
	Moreover, we have
	\begin{align}
	&\mathbb{L}(P_{XY},2)\leq \!H(Q^*)\!+\frac{\log e}{e}+\!2,\label{koonwi11}\\
	&\mathbb{L}(P_{XY},|\mathcal{X}|)\leq \!H(Q^*)\!+\frac{1+\log e}{2}+\!1+\!\ceil{\log (|\mathcal{X}|)}\label{koonwi22},
	\end{align}
	where $Q^*$ is defined in Theorem \ref{th1}.
	Finally, for any $P_{X\mathcal{C}'}$ (not necessarily $P_{X\mathcal{C}'}\in\hat{\mathcal{P}}_{X\mathcal{C}'}$) with $|\mathcal{C}'|\leq |\mathcal{X}|$ we have
	\begin{align}
	\mathbb{L}(P_{X\mathcal{C}'},|\mathcal{C}'|)\leq \ceil{\log{|\mathcal{C}'|}}.\label{kos3} 
	\end{align}
\end{theorem}
 \begin{proof}
 	The proof is based on two-part construction coding and is similar to Theorem \ref{th1} and \cite[Theorem 2]{zero}. To achieve \eqref{log}, we use the solution to $h_0(P_{X\mathcal{C}'})=g_0(P_{X\mathcal{C}'})$ instead of the greedy entropy-based algorithm. Moreover, to achieve \eqref{ass} and \eqref{mass}, we use two-part construction coding and inequalities obtained in \cite[Theorem 1]{zero}. Upper bounds \eqref{koonwi11} and \eqref{koonwi22} are obtained in Theorem \ref{th1}. Finally, to achieve \eqref{kos3}, let the shared key $W$ be independent of $(X,Y)$ and has uniform distribution $\{1,,\ldots,T\}=\{1,\ldots,|\mathcal{C}'|\}$. We construct $\tilde{C}$ using one-time pad coding. We have
 	\begin{align*}
 	\tilde{C}=\mathcal{C}'+W\ \text{mod}\ |\mathcal{Y}|,
 	\end{align*}
 	where $\mathcal{C}'$ is defined in \eqref{cache1} and clearly we have
 	\begin{align}
 	I(\tilde{C};X)=0.
 	\end{align}
 	Then, $\tilde{C}$ is encoded using any lossless code which uses at most $\ceil{\log(|\mathcal{C}|)}$ bits. At decoder side, we first decode $\mathcal{C}'$ using the shared key. We then decode each demanded file by using the cache contents and $\mathcal{C}'$.
 \end{proof}
\begin{remark}
	Clearly, the upper bound obtained in \eqref{kos3} improves the bounds in Theorem \ref{th1}. Since, when $|\mathcal{C}'|\leq |\mathcal{X}|$ we have
	\begin{align}
	\ceil{\log{|\mathcal{C}'|}}\leq \!H(Q^*)\!+\frac{1+\log e}{2}+\!1+\!\ceil{\log (|\mathcal{X}|)}. 
	\end{align} 
\end{remark}
\begin{figure}[]
	\centering
	\includegraphics[scale = .4,trim={0 9cm 0 9cm}]{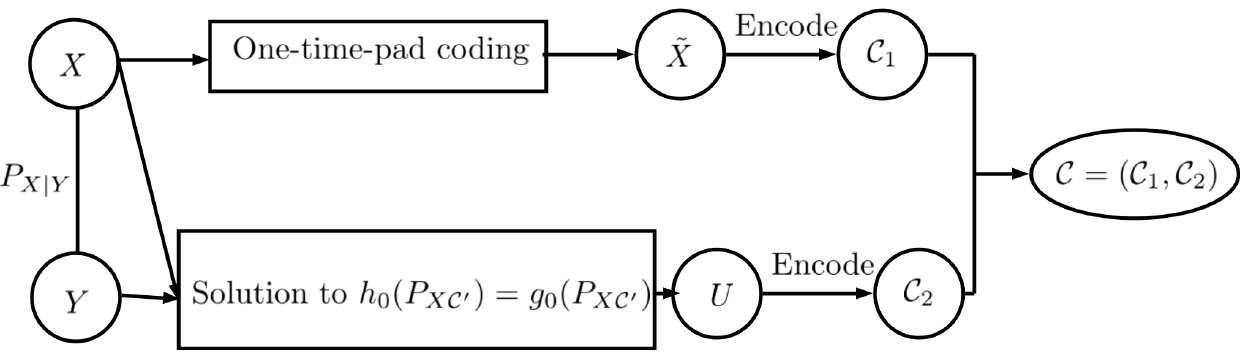}
	\caption{In this work, we use two-part construction coding strategy to send codewords over the channels. We hide the information of $X$ using one-time-pad coding and we then use the solution of $h_0(P_{X\mathcal{C}'})=g_0(P_{X\mathcal{C}'})$ to construct $U$.} 
	\label{kesh11}
\end{figure}
Next, we provide a numerical example that shows \eqref{log} improves \eqref{koonwi}.
\begin{example}
	Let $P_{X|\mathcal{C}'}=\begin{bmatrix}
	1 &1 &1 &0 &0 &0\\
	0 &0 &0 &1 &1 &1
	\end{bmatrix}$ and $P_{\mathcal{C}'}=[\frac{1}{8},\frac{2}{8},\frac{3}{8},\frac{1}{8},\frac{1}{16},\frac{1}{16} ]$. Clearly, in this case $X$ is a deterministic function of $\mathcal{C}'$. Using the linear program proposed in \cite{borz}, we obtain a solution as $P_{\mathcal{C}'|u_1}=[0.75,0,0,0.25,0,0]$, $P_{\mathcal{C}'|u_2}=[0,0.75,0,0.25,0,0]$, $P_{\mathcal{C}'|u_3}=[0,0,0.75,0,0.25,0]$, $P_{\mathcal{C}'|u_4}=[0,0,0.75,0,0,0.25]$ and $P_U=[\frac{1}{6},\frac{1}{3},\frac{1}{4},\frac{1}{4} ]$ which results $H(U)=1.9591$ bits. We have $H(U)=\mathcal{K}(P_{XY})\leq 1.9591$. Moreover, we have
	\begin{align*}
	P_{\mathcal{C}'|X}=\begin{bmatrix}
	&\frac{1}{6} &0\\ &\frac{1}{3} &0\\ &\frac{1}{2} &0\\ &0 &\frac{1}{2}\\&0 &\frac{1}{4}\\&0 &\frac{1}{4}
	\end{bmatrix}.
	\end{align*}
	Using the greedy search algorithm we have $P_{Q*}=[\frac{1}{2}\ \frac{1}{4}\ \frac{1}{6}\ \frac{1}{12}]$, hence, $H(Q^*)=1.7296$. Thus, 
	\begin{align*}
	H(Q^*)+\frac{\log e}{e}=2.2596\geq \mathcal{K}(P_{XY})=1.9591.
	\end{align*}
\end{example} 
 \section{conclusion}
 We have studied a cache-aided compression problem with a perfect privacy constraint, where the information delivered over the shared link during the delivery phase is independent of $X$ that is correlated with the files in the database that can be requested by the users. We have strengthened the previous achievable scheme by using a greedy entropy-based algorithm instead of the FRL. The greedy algorithm aims to solve the minimum entropy functional representation and is optimal within a constant gap. Considering two special cases the obtained bounds are strengthened. Specifically, when the size of the private data is large, we need significantly less bits to send over the channel with less shared key size.  
\section{acknowledgment}
The authors would like to express their gratitude to Yanina Shkel for suggesting the minimum entropy coupling problem, the greedy entropy-based approach, and providing related references.
	\bibliographystyle{IEEEtran}
	\bibliography{IEEEabrv,IZS}
\end{document}